\documentclass[conference]{IEEEtran}

\usepackage{cite}
\usepackage{amsmath,amssymb,amsfonts}
\usepackage{graphicx}
\usepackage{textcomp}
\usepackage{xcolor}
\usepackage[utf8]{inputenc} 
\usepackage[T1]{fontenc}
\usepackage{url}
\usepackage{ifthen}
\usepackage[english]{babel}
\addto\captionsenglish{}
\usepackage{mathtools}
\usepackage{booktabs}
\usepackage{amsfonts}
\usepackage{bbm}

\usepackage{amsthm}
\newtheorem{theorem}{Theorem}
\newtheorem{lemma}{Lemma}

\newtheorem*{example}{Example}
\usepackage[colorlinks=true,bookmarks=false,citecolor=blue,urlcolor=blue]{hyperref} %
\usepackage[
	ruled,
	vlined,
	boxed, 
	linesnumbered, 
	commentsnumbered]{algorithm2e}
\usepackage{float}
\usepackage{siunitx}
\usepackage{placeins}
\usepackage{tikz}
\usepackage{pgfplots}
\usepackage{comment}
\usepackage{acronym}
\usepackage{makecell}
\usepackage{braket}

\definecolor{A1}{RGB}{146,84,3}%
\definecolor{A2}{RGB}{8,51,110}%
\definecolor{A3}{rgb}{0,.59,.51}%
\definecolor{A4}{rgb}{.63,.13,.13}%
\definecolor{A5}{RGB}{160,0,120}%

\definecolor{kit-green100}{rgb}{0,.59,.51}%
\definecolor{kit-green70}{rgb}{.3,.71,.65}
\definecolor{kit-green50}{rgb}{.50,.79,.75}
\definecolor{kit-green30}{rgb}{.69,.87,.85}
\definecolor{kit-green15}{rgb}{.85,.93,.93}
\definecolor{KITgreen}{rgb}{0,.59,.51}
\definecolor{KITblue}{rgb}{.27,.39,.66}
\definecolor{KITorange}{rgb}{.87,.60,.10}
\definecolor{myblue}{RGB}{5,113,176}
\definecolor{myred}{RGB}{202,0,32}
\definecolor{mygreen}{RGB}{0,136,55}
\definecolor{KITpurple}{RGB}{160,0,120}
\definecolor{KITyellow}{rgb}{.78,.59,.20}
\definecolor{QC1}{RGB}{170,220,224}
\definecolor{QC2}{RGB}{114,188,213}
\definecolor{QC3}{RGB}{82,143,173}
\definecolor{QC4}{RGB}{55,103,149}
\definecolor{QC5}{RGB}{30,70,110}

\definecolor{EG1}{RGB}{255,230,163}
\definecolor{EG2}{RGB}{247,170,88}
\definecolor{EG3}{RGB}{231,98,84}
\definecolor{EG4}{RGB}{149,72,75}
\definecolor{EG5}{RGB}{114,100,158}

\usepackage{multirow}
\usepackage{bm}
\renewcommand{\vec}[1]{\bm{#1}}

\usepackage{pifont}%

\usepackage{bm}
\renewcommand{\vec}[1]{\bm{#1}}
\usepackage{bbm}

\newcommand\blfootnote[1]{%
  \begingroup
  \renewcommand\thefootnote{}\footnote{#1}%
  \addtocounter{footnote}{-1}%
  \endgroup
}

\binoppenalty=\maxdimen
\relpenalty=\maxdimen

\acrodef{BCH}{Bose--Chaudhuri--Hocquenghem}
\acrodef{FER}{frame error rate}
\acrodef{BSC}{binary symmetric channel}
\acrodef{DE}{density evolution}
\acrodef{LDPC}{low-density parity-check}
\acrodef{QLDPC}{quantum low-density parity-check}
\acrodef{VN}{variable node}
\acrodef{CN}{check node}
\acrodef{BP}{belief propagation}

\acrodef{GB}{generalized bicycle}
\acrodef{PCM}{parity check matrix}
\acrodefplural{PCM}[PCMs]{parity check matrices}
\acrodef{OSD}{ordered statistics decoding}
\acrodef{QEC}{quantum error correction}
\acrodef{QSC}{quantum stabilizer code}
\acrodef{CSS}{Calderbank–Shor–Steane}
\acrodef{NN}{neural network}
\acrodef{LLR}{log-likelihood ratio}
\acrodef{MWPM}{minimum weight perfect matching}
\acrodef{MS}{min-sum}
\acrodef{FG}{finite geometry}
\acrodef{EG}{Euclidean geometry}
\acrodef{PG}{projective geometry}

\acrodefplural{FG}[FGs]{finite geometries}
\acrodefplural{EG}[EGs]{Euclidean geometries}
\acrodefplural{PG}[PGs]{projective geometries}

\acrodef{QC}{quasi-cyclic}
\acrodef{CPM}{circulant permutation matrix}
\acrodef{ML}{maximum likelihood}
\acrodef{CAMEL} {\textbf{C}ycle \textbf{A}ssembling and \textbf{M}itigating \textbf{E}nsemb\textbf{L}e decoding)}

\definecolor{cR1}{rgb}{0,.59,.51}
\definecolor{cR2}{RGB}{162,34,35}
\definecolor{cR3}{RGB}{217,115,14}
\definecolor{cR4}{RGB}{117,75,34}

\pgfplotsset{compat=1.18}
\begin{document}

\title{A Joint Code and Belief Propagation Decoder Design for Quantum LDPC Codes}

 \author{%
   \IEEEauthorblockN{Sisi Miao, Jonathan~Mandelbaum, Holger~Jäkel, and Laurent Schmalen}
   \IEEEauthorblockA{Karlsruhe Institute of Technology (KIT),
Communications Engineering Lab (CEL),
76187 Karlsruhe, Germany\\
Email: {\{\texttt{firstname.lastname@kit.edu\}}}}\vspace{-1em}}

\maketitle

\begin{abstract}
Quantum low-density parity-check (QLDPC) codes are among the most promising candidates for future quantum error correction schemes. However, a limited number of short to moderate-length QLDPC codes have been designed and their decoding performance is sub-optimal with a quaternary belief propagation (BP) decoder due to unavoidable short cycles in their Tanner graphs. In this paper, we propose a novel joint code and decoder design for QLDPC codes. The constructed codes have a minimum distance of about the square root of the block length. In addition, it is, to the best of our knowledge, the first QLDPC code family where BP decoding is not impaired by short cycles of length 4. This is achieved by using an ensemble BP decoder mitigating the influence of assembled short cycles.  We outline two code construction methods based on classical quasi-cyclic codes and finite geometry codes. Numerical results demonstrate outstanding decoding performance over depolarizing channels. 
\end{abstract}
\begin{IEEEkeywords}
Quantum error correction, LDPC codes, belief propagation decoding
\end{IEEEkeywords}
\vspace{-1em}

\section{Introduction}\label{sec:introduction}
\blfootnote{This work has received funding from the European Research Council (ERC) under the European Union’s Horizon 2020 research and innovation programme (grant agreement No. 101001899) and the German Federal Ministry of Education and Research (BMBF) within the project Open6GHub (grant agreement 16KISK010).}
\Acp{QLDPC} codes are among the most promising candidates for future \ac{QEC} schemes~\cite{Gottesman14fault}. Various promising code constructions have been proposed, e.g., ~\cite{mackay2004sparse,Hagiwara2007QQLDPC,Aly08QEG,tillich2013quantum,panteleev2022asymptotically}. Still, some challenges remain. First, only a small number of good short to moderate-length QLDPC codes have been constructed, which are of special interest due to their low implementation complexity and their low decoding latency. Second, most construction methods ignore the structure of the underlying Tanner graph of the constructed codes. Thus, unavoidable short cycles of length 4 significantly impair the quaternary \ac{BP} decoding performance. However, the good decoding performance for classical LDPC codes and low decoding latency make BP an attractive candidate for \ac{QEC}. Therefore, many previous works have tried to improve the decoding performance in the presence of short cycles, e.g., by modifying the BP decoder~\cite{lai2021log,Babar15, panteleev2021degenerate,miao2023quaternary}, or by introducing post-processing steps such as \ac{OSD}~\cite{panteleev2021degenerate, roffe2020decoding}. However, these approaches cannot guarantee to completely mitigate the influence of the short cycles. Furthermore, the extra decoding latency introduced by these methods makes them less appealing for \ac{QEC} where linear or even constant time decoding complexity is desired to achieve ultra-low latency decoding. For example, the complexity of \ac{OSD} is $\mathcal{O}\left(n^3\right)$ while the complexity of BP decoding is only $\mathcal{O}\left(n\right)$.

In this work, we construct short to moderate-length \ac{QLDPC} codes with a good decoding performance with only BP decoding, achieved by joint code and decoder design such that the proposed ensemble \ac{BP} decoding is not impaired by short cycles of length 4. The proposed scheme is referred to as \textbf{CAMEL} (\textbf{C}ycle \textbf{A}ssembling and \textbf{M}itigating with \textbf{E}nsemb\textbf{L}e decoding). We introduce two exemplary code construction methods in this paper. The first one is constructed from classical \ac{QC} codes. The second one reuses the construction in~\cite{Aly08QEG} which is based on \acp{FG}. We evaluate the performance of CAMEL with numerical simulations over depolarizing channels.

\textbf{Notation}: Boldface letters denote vectors and matrices, e.g., $\vec{a}$ and $\vec{A}$. The $i$-th component of vector $\vec{a}$ is denoted by $a_i$, and the element at the $i$-th row and $j$-th column of $\vec{A}$ is denoted by $A_{i,j}$. Let $\vec{A}_{i,:}$ be the $i$-th row of a matrix $\vec{A}$ and $\vec{A}_{:,i}$ the $i$-th column of $\vec{A}$. $\vec{A}^{\mathsf{T}}$ denotes the matrix transpose. The set $\{0,1,2,\cdots,p-1\}$ is denoted by $[p]$ for any $p\in \mathbb{N}$. The trace inner product for $x,y\in \mathbb{F}_4$ is written as $\langle x,y\rangle \in \{0,1\}$. It evaluates to $1$ if $x\neq 0, y\neq 0$ and $x\neq y$, and $0$ otherwise. We use $\oplus$ to denote binary summation. The indicator function is denoted by $\mathbbm{1}_{\{\cdot\}}$. Throughout the paper, the indexing of vector and matrix elements always starts from $0$.

\section{Preliminaries}\label{sec:preliminaries}
We consider a \ac{CSS} type \ac{QSC}~\cite{calderbank1998quantum} given by Theorem~\ref{theorem:css}.

\begin{theorem}\label{theorem:css}
      Consider two classical binary linear codes $\mathcal{C}_1$ and $\mathcal{C}_2$ with parameters $[n,k_1,d_1]$  and $[n,k_2,d_2]$. Their \acp{PCM} are $\vec{H}_X\in \mathbb{F}_2^{(n-k_1)\times n}$ and $\vec{H}_Z\in \mathbb{F}_2^{(n-k_2)\times n}$, respectively. If $\mathcal{C}_2^{\perp}\subseteq \mathcal{C}_1$, i.e., satisfying the so-called twisted condition
\begin{equation}
    \vec{H}_X \vec{H}_Z^{\mathsf{T}} = \vec{0},
    \label{eq:symplectic_criterion_css}
\end{equation} an $[[n,k_1+k_2-n,\min\{d_1,d_2\}]]$ \ac{QSC} can be constructed.
\end{theorem}
Unless mentioned differently, we restrict ourselves to the case where ${k_1=k_2=:k}$. Thus, the check matrix of the \ac{QSC} is written as
\begin{equation}
    \vec{S} = \left(\begin{array}{c}
\omega \vec{H}_X    \\
\bar{\omega}\vec{H}_Z
\end{array}
\right)\in \mathbb{F}_4^{2(n-k)\times n}
    \label{eq:css_quarternary}
\end{equation} with $\omega$ and $\bar{\omega}$ being elements of the Galois field $\mathbb{F}_4=\{0,1,\omega,\bar{\omega}\}$.
Furthermore, $\vec{H}_X$ and $\vec{H}_Z$ being sparse matrices results in a \ac{QLDPC} code. 

To estimate the error $\bm{e}\in \mathbb{F}_4^n$ that occurred, the syndrome ${\bm{z}\in \mathbb{F}_2^{2(n-k)}}$ is measured on the stabilizer generators. For simulation purpose, the syndrome $\bm{z}$ is computed as $z_j = {\bigoplus_{i\in [n]}} \langle e_{i}, S_{j,i} \rangle$ with $j\in [2(n-k)]$. In this work, we perform quaternary \ac{BP} decoding~\cite{DM98} on the Tanner graph associated with the check matrix $\vec{S}$. A Tanner graph is a bipartite graph with two sets of vertices: the \acp{VN} corresponding to the code symbols and the \acp{CN} corresponding to the checks and thus to rows of $\bm{S}$. A VN $\mathsf{v}_i$ is connected to a CN $\mathsf{c}_j$ if the corresponding entry $S_{j,i}\neq 0$. 

Note that \eqref{eq:symplectic_criterion_css} requires an even number of overlapping ones between any row of $\bm{H}_X$ and any row of $\bm{H}_Z$. Therefore, the Tanner graph associated with $\bm{S}$ has a girth of either infinity or $4$. First, if the Tanner graph has a girth of infinity, i.e., it is a tree, then the Tanner graphs associated with $\bm{H}_X$ and $\bm{H}_Z$ are necessarily trees. Such binary codes are known to have poor minimum distance~\cite{ETV99}. Hence, short cycles are necessary in constructing good QLDPC codes. However, the Tanner graph structure needs careful optimization for BP decoding.
One potential solution is to construct component matrices $\bm{H}_{X}$ and $\bm{H}_{Z}$ with a girth of at least $6$ yet permitting 4-cycles between $\bm{H}_{X}$ and $\bm{H}_{Z}$, as, e.g., done in~\cite{Hagiwara2007QQLDPC}. Therefore, acceptable performance can be attained by decoding the $\bm{X}$ and $\bm{Z}$ errors separately using two binary \ac{BP} decoders operating on the \acp{PCM} $\bm{H}_{Z}$ and $\bm{H}_{X}$, respectively. However, binary decoding ignores the correlation between $\bm{X}$ and $\bm{Z}$ errors and is thus inherently sub-optimal, see, e.g.,~\cite{lai2021log,panteleev2021degenerate,miao2023quaternary} for a performance comparison.
Hence, to take the correlation into account while still mitigating the influence of short cycles, 
we propose CAMEL, a novel code construction where an ensemble of quaternary BP decoding is not impaired by short cycles of length 4.

\section{CAMEL: Joint Code and Decoder Design}\label{sec:code_and_decoder}

 CAMEL consists of the construction of codes where all short cycles are assembled onto a single \ac{VN} and an ensemble BP decoder in which the influence of short cycles is fully mitigated.

To this end, we first construct two classical binary LDPC codes whose parity check matrices $\vec{H}_1,\vec{H}_2\in \mathbb{F}_2^{m\times n}$ fulfill
\begin{equation}
    \vec{H}_1\vec{H}_2^{\mathsf{T}} = \vec{1}_{m\times m},
    \label{eq:HHT1}
\end{equation}
where $\vec{1}_{m\times m}$ denotes an all-one matrix of size $m\times m$. A straightforward way to satisfy \eqref{eq:HHT1} is to construct matrices $\vec{H}_1$ and $\vec{H}_2$ such that any row of $\vec{H}_1$ overlaps with any row of $\vec{H}_2$ in exactly one position. It is possible to construct $\vec{H}_1$ and $\vec{H}_2$ fulfilling \eqref{eq:HHT1} such that the Tanner graph associated with the matrix $\begin{pmatrix}
    \vec{H}_1^{\mathsf{T}}&\vec{H}_2^{\mathsf{T}}
\end{pmatrix}^{\mathsf{T}}$ has a girth of at least $6$.
Two explicit code constructions will be presented in the following sections.

Next, we construct two new \acp{PCM} by appending an all-one column vector $\vec{1}_{m\times 1}$ to $\vec{H}_1$ and $\vec{H}_2$, respectively, resulting in
\begin{equation*}
    \vec{H}_{X} = \begin{pmatrix}\vec{H}_1 \mid \vec{1}_{m\times 1}\\\end{pmatrix} \in \mathbb{F}_2^{m\times (n+1)}
\end{equation*}
and 
\begin{equation*}
    \vec{H}_{Z} = \begin{pmatrix}\vec{H}_2 \mid \vec{1}_{m\times 1}\\\end{pmatrix} \in \mathbb{F}_2^{m\times (n+1)}.
\end{equation*}
Then, we can verify that
\[
\vec{H}_X\vec{H}_Z^{\mathsf{T}} = \begin{pmatrix}\vec{H}_1 \mid \vec{1}\\\end{pmatrix} \begin{pmatrix}\vec{H}_2^{\mathsf{T}} \\ \vec{1}^{\mathsf{T}} \end{pmatrix} = \vec{H}_1\vec{H}_2^{\mathsf{T}} + \vec{1}_{m\times m} = \vec{0}
\]
and \eqref{eq:symplectic_criterion_css} is fulfilled. Note that all cycles of length $4$ are nested on $\mathsf{v}_n$. Then, using Theorem~\ref{theorem:css}, we obtain an $[[n+1, n+1-\text{rank}(\bm{H}_X)-\text{rank}(\bm{H}_Z),d]]$ QLDPC code. 

\begin{figure}
    \centering
    		\tikzset{line/.style={-latex}} 
		\begin{tikzpicture}
    \node[draw=none] at (0,3.2) {\footnotesize$\vec{z}$};
    \node[draw,circle,fill,inner sep=1pt] at (0.3,3)(s) {};
    \node[draw,circle,fill,inner sep=1pt] at (0.3,3.3) (s1){};
    \node[draw,circle,fill,inner sep=1pt] at (0.3,2.7) (s2) {};

    \draw[] (-0.2,3) -- (s);
    
    \draw[] (2,3.7) rectangle ++(1,0.4);
    \node[draw=none] at (1,4.1) {\footnotesize $v_n=0$};
    \node[draw=none] at (2.5,3.9)(BP1) {\footnotesize BP};  
    \draw[line] (s1) |- (2,3.9);
    
    \draw[] (2,3.1) rectangle ++(1,0.4);
    \node[draw=none] at (1,3.5) {\footnotesize $v_n=1$};
    \node[draw=none] at (2.5,3.3)(BP2) {\footnotesize BP};
    \draw[line] (s) |- (2,3.3);
    
    \draw[] (2,2.5) rectangle ++(1,0.4);
    \node[draw=none] at (1,2.9) {\footnotesize $v_n=\omega$};
    \node[draw=none] at (2.5,2.7)(BP3) {\footnotesize BP};
    \draw[line] (s) |- (2,2.7);
    \draw[] (2,1.9) rectangle ++(1,0.4);
    \node[draw=none] at (1,2.3) {\footnotesize $v_n=\bar{\omega}$};
    \node[draw=none] at (2.5,2.1)(BP4) {\footnotesize BP};
    \draw[line] (s2) |- (2,2.1);

    \node[draw,rectangle,text width=0.4cm, minimum height=1cm,align=center] at (4.5,3)(w) {\footnotesize ML};

    \node[draw=none] at (3.5,4.1) {\footnotesize $\hat{\vec{e}}_0$};
    \node[draw=none] at (3.5,3.5) {\footnotesize $\hat{\vec{e}}_1$};
    \node[draw=none] at (3.5,2.9) {\footnotesize $\hat{\vec{e}}_2$};
    \node[draw=none] at (3.5,2.3) {\footnotesize $\hat{\vec{e}}_3$};

    \draw[line] (3,3.9) -| (w);
    \draw[line] (3,3.3) -- (4.2,3.3);
    \draw[line] (3,2.7) -- (4.2,2.7);
    \draw[line] (3,2.1) -| (w);
    
     \node[draw=none] at (5.5,3.2) {\footnotesize$\hat{\vec{e}}$};
     \draw[line] (w) -- (6,3);
     \end{tikzpicture}
    \caption{Block diagram of the proposed ensemble decoder.}
    \vspace{-1em}
    \label{fig:block_diagramm}
\end{figure}
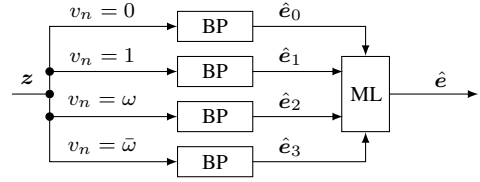

We proceed by introducing the decoder part of CAMEL which essentially relies on ensemble decoding~\cite{MBBP}. As depicted in Fig.~\ref{fig:block_diagramm}, after measuring the syndrome $\vec{z}$, four BP decodings are performed in parallel. 
In each BP decoding, the last bit is assigned a distinct fixed value $\eta\in \mathbb{F}_4$. The effect of doing so will be discussed below.

Let $\hat{\vec{e}}_i$ denote the error estimate of the BP decoding in the ${i\text{-th}}$ path of the decoder shown in Fig.~\ref{fig:block_diagramm} and define ${\mathcal{I}:=\{{i\in [4]: \hat{\vec{e}}_i \text{ satisfies syndrome }\bm{z}}\}}$. If $|\mathcal{I}|=0$, the decoder declares a decoding failure.
Otherwise, we perform a \textit{\ac{ML}-in-the-list} step by choosing the error candidate $\hat{\vec{e}}=\hat{\vec{e}}_{i^{\star}}$, ${i^{\star}\in\mathcal{I}}$, of lowest weight, i.e.,
$i^{\star} = \text{arg } \underset{i\in \mathcal{I}}{\text{min}} \text{ w} (\hat{\vec{e}}_i)$ where $\text{w}(\cdot)$ denotes the number of nonzero elements in the argument.

The approach of decoding using multiple decoders with hard guesses for certain bits is known as \textit{decimation} in classical coding theory, e.g., in~\cite{Vahid15}.
Next, we analyze the influence of decimating one bit in BP decoding for QLDPC codes.
To this end, we consider a quarternary BP decoder passing probabilities.
One can verify that the upcoming conclusions also hold for BP decoders with \ac{LLR} message passing~\cite{DM98,DF07} and their refined version~\cite{lai2021log}.

Assume that we provide the decoder with the hard guess $\eta$ of the last \ac{VN}, i.e., $\mathsf{v}_n=\eta$.
For an arbitrary decoding iteration, the outgoing message vector of $\mathsf{v}_n$ to a neighboring \ac{CN} $\mathsf{c}_j$ is $\begin{pmatrix}
    m_{\mathsf{v},n\rightarrow j}^{(0)}&m_{\mathsf{v},n\rightarrow j}^{(1)}&m_{\mathsf{v},n\rightarrow j}^{(\omega)}&m_{\mathsf{v},n\rightarrow j}^{(\bar{\omega})}
\end{pmatrix}$
where 
\begin{equation}\label{eq:VNmsg1}
    m_{\mathsf{v},n\rightarrow j}^{(a)} = b\cdot P_{\mathrm{ch}}(\mathsf{v}_n=a)\prod_{j'\in \mathcal{M}(n)\backslash\{j\}} m_{\mathsf{c},j'\rightarrow n}^{(a)}
\end{equation}
for $a\in \mathbb{F}_4$ with $\mathcal{M}(i)$ denoting the set of indices of the neighboring CNs of VN $\mathsf{v}_i$.
The parameter $b$ is a normalization factor such that the probabilities sum up to $1$. As we provide a hard guess of $\mathsf{v}_n=\eta$, the channel probability is $P_{\mathrm{ch}}(\mathsf{v}_n=a)=\mathbbm{1}_{\{a = \eta\}}$. Therefore, the outgoing message of $\mathsf{v}_n$ in \eqref{eq:VNmsg1} evaluates to
\begin{equation}\label{eq:VNmsg}
    m_{\mathsf{v},n\rightarrow j}^{(a)}=\mathbbm{1}_{\{a = \eta\}},
\end{equation}
regardless of the incoming messages from the \acp{CN}.
We now inspect the outgoing message of a \ac{CN} $\mathsf{c}_j$ that is a neighbor of the decimated \ac{VN} $\mathsf{v}_n$, i.e., $j\in\mathcal{M}(n)$:
\begin{equation}\label{eq:CNmsg1}
    m_{\mathsf{c},j\rightarrow i}^{(a)} = c\cdot \sum_{\bm{t}:t_i=a}f_{\bm{S}_{j,:}}(\bm{t},z_j) \prod_{i'\in \mathcal{N}(j)\backslash\{i\}}m_{\mathsf{v},i'\rightarrow j}^{(t_{i'})}
\end{equation}
where $c$ is again some normalization factor and $\mathcal{N}(j)$ denotes the set of indices of the neighboring VNs of $\mathsf{c}_j$. The check function $f_{\bm{S}_{j,:}}(\bm{t},z_j)$ indicates if an error vector $\bm{t}$ fulfills the syndrome $z_j$ specified by the $j$-th row of the check matrix $\bm{S}_{j,:}$, i.e., , $f_{\bm{S}_{j,:}}(\bm{t},z_j)=1$ if ${\bigoplus}_{i\in \mathcal{N}(j)} \langle t_{i}, S_{j,i} \rangle= z_j$ and $0$ otherwise. Together with \eqref{eq:VNmsg}, we derive that for any $i\in \mathcal{N}(j), i\neq n$, the outgoing message in \eqref{eq:CNmsg1} is
\begin{align}
    &m_{\mathsf{c},j\rightarrow i}^{(a)} =  
   c \cdot \sum_{\bm{t}:t_i=a,t_n=\eta} f_{\bm{S}_{j,:}}(\bm{t},z_j) \,\,\,\,\,\,\,\,\,\,\,\,\,\,\,\mathclap{\prod_{i'\in \mathcal{N}(j)\backslash\{i,n\}}} \;\,\,\,\;\,\,\,\,\,m_{\mathsf{v},i'\rightarrow j}^{(t_{i'})}\nonumber\\
   &=c\cdot \;\,\,\,\mathclap{\sum_{\bm{t}_{\sim n}:t_i=a}}\;\,\,\,\,\langle \eta,S_{j,n}\rangle \oplus  f_{\bm{S}_{j,\sim n}}(\bm{t}_{\sim n},z_j) \;\,\,\,\;\,\,\,\,\,\,\,\,\mathclap{\prod_{i'\in \mathcal{N}(j)\backslash\{i,n\}}} \;\,\,\,\;\,\,\,\,\, m_{\mathsf{v},i'\rightarrow j}^{(t_{i'})},
   \label{eq:decimation}
\end{align}
where $\bm{t}_{\sim n}$ and $\bm{S}_{j,\sim n}$ denote $\bm{t}$ and $\bm{S}_{j,:}$ excluding their last entry, i.e., $t_n$ and $S_{j,n}$, respectively.

From \eqref{eq:decimation}, it follows that messages associated with $\mathsf{v}_n$ are excluded in \ac{BP}, except for $\langle \eta,S_{j,n}\rangle$, which is a constant. Thus, $\mathsf{v}_n$ can be removed together with all edges incident to it. This essentially dissolves all short cycles of length 4 in each decoding path as they are all assembled on $\mathsf{v}_n$. Hence, their influence is fully mitigated.

\section{Quasi-Cyclic QLDPC Codes}\label{sec:coded_from_QC}

In \cite{Hagiwara2007QQLDPC}, a class of \ac{QC} QLDPC codes without short cycles of length $4$ in their component matrices $\bm{H}_X$ and $\bm{H}_Z$ were constructed. In this section, we adapt this method to construct codes fulfilling (\ref{eq:HHT1}). We first introduce a few definitions.

Let $p$ be a prime number and $\mathbb{F}_p$ be the prime field. We focus on a class of classical binary \ac{QC} LDPC codes defined as the null space of a \ac{PCM}
\begin{equation}
\setlength\arraycolsep{2pt}
    \bm{H} = \begin{pmatrix}
        I(c_{0,0})&I(c_{0,1})&\cdots&I(c_{0,L-1})\\
        I(c_{1,0})&I(c_{1,1})&\cdots&I(c_{1,L-1})\\
        \vdots&\vdots&\ddots&\vdots\\
        I(c_{J-1,0})&I(c_{J-1,1})&\cdots&I(c_{J-1,L-1})\\
    \end{pmatrix}\in \mathbb{F}_2^{Jp\times Lp},
\end{equation}
obtained by replacing each element $c_{i,j}\in [p]$ of a base matrix $\bm{\mathcal{H}}\in[p]^{J\times L}$ by a \ac{CPM} $I(c_{i,j})$. A \ac{CPM} $I(x)$ is obtained by cyclically right shifting all the rows of the identity matrix $\bm{I}\in \mathbb{F}_2^{p\times p}$ by $x$ positions. For brevity, we denote the procedure of obtaining a binary PCM from a base matrix $\bm{\mathcal{H}}$ as $\bm{H}=$ Cyc($\bm{\mathcal{H}}$). Note that the commutative group of \acp{CPM} $I(x)$ with ${x\in[p]}$ under matrix multiplication is isomorphic to the additive group of $\mathbb{F}_p$ as $I(x)I(y) = I((x+y)\text{ mod }p)$. Furthermore, for a commutative group $\mathcal{G}$ and a subgroup ${\mathcal{G}'}$ of $\mathcal{G}$, the coset of ${g\in \mathcal{G}}$ w.r.t. $\mathcal{G}'$ is defined as
$[g]_{\mathcal{G}'}:=\{gh:h\in {\mathcal{G}'}\}$. Note that two cosets are either identical or disjoint.

Following the notation of \cite{Hagiwara2007QQLDPC}, a vector $\bm{x}\in \mathbb{F}_p^L$ is called \textit{multiplicity odd} if every element of $\bm{x}$ occurs an odd number of times and \textit{multiplicity free} if every element of $\bm{x}$ is unique. The vector $\bm{x}$ is a \textit{permutation vector} if it contains all elements of $\mathbb{F}_p$ exactly once. Hence, a permutation vector of length $p$ is multiplicity odd and free. Next, we need Lemma~\ref{lemma:condition_for_HHT1}.

\begin{lemma}\label{lemma:condition_for_HHT1} Consider two matrices ${\bm{\mathcal{A}}\in \mathbb{F}_p^{J_1\times L}}$ and $\bm{\mathcal{B}}\in \mathbb{F}_p^{J_2\times L}$. Then, 
$
\mathrm{Cyc}(\bm{\mathcal{A}})\mathrm{Cyc}(\bm{\mathcal{B}})^{\mathsf{T}} =\bm{1}_{ J_1p\times J_2p}
$
if and only if the difference vector between any two rows of $\bm{\mathcal{A}}$ and $\bm{\mathcal{B}}$ is multiplicity odd and contains all elements of $\mathbb{F}_p$.
\end{lemma}

\begin{proof}
Let $\bm{a}$ be an arbitrary row of $\bm{\mathcal{A}}$ and $\bm{b}$ be a row of $\bm{\mathcal{B}}$. First, we have:
\begin{align}
    \mathrm{Cyc}(\bm{a})\mathrm{Cyc}(\bm{b})^{\mathsf{T}}=\sum_{i=0}^{L-1}I(a_i)I(b_i)^{\mathsf{T}}\nonumber
    =\sum_{i=0}^{L-1}I(a_i-b_i).\label{lemma:proof:sumcpms}
\end{align}
This is a summation of $L$ CPMs. Note that the ones in two CPMs either completely overlap or do not overlap at all. To ensure that $\sum_{i=0}^{L-1}I(a_i-b_i)=\bm{1}_{p\times p}$, the set $\{a_i-b_i:i \in [L]\}$ must contain all elements of $\mathbb{F}_p$ an odd number of times. Thus, the difference vector of any two rows $\bm{a}$ and $\bm{b}$ must be multiplicity odd and contain all elements of $\mathbb{F}_p$ at least once. As this holds for any two rows, it concludes the proof.\end{proof}

In order to obtain an all-one matrix, Lemma~\ref{lemma:condition_for_HHT1} implies $L\geq p$.
Yet $L> p$ implies a Tanner graph with girth $4$ \cite{Hagiwara2007QQLDPC}. Hence, we choose $L=p$.
 In this case, the set $\{a_i-b_i:i \in [L]\}$ yields a permutation vector and a check matrix with girth at least $6$ can be constructed, as shown by Theorem \ref{theorem:QCcosetconstruction}.

\begin{theorem}\label{theorem:QCcosetconstruction}
Let $p$ be a prime number. There exists a base matrix $\bm{\mathcal{H}}\in \mathbb{F}_p^{\ell \times p}$ with $\ell\leq p-1$ such that any partition of $\bm{\mathcal{H}}$ into ${\bm{\mathcal{H}}=\begin{pmatrix}
        \bm{\mathcal{H}}_1^{\mathsf{T}}&
        \bm{\mathcal{H}}_2^{\mathsf{T}}
    \end{pmatrix}}^{\mathsf{T}}$ yields $\mathrm{Cyc}( \bm{\mathcal{H}}_1)\mathrm{Cyc}( \bm{\mathcal{H}}_2)^{\mathsf{T}}=\bm{1}$. Besides, the Tanner graph of $\mathrm{Cyc}( \bm{\mathcal{H}})$ has girth at least 6.
\end{theorem}
\begin{proof}
    The proof is constructive. Let $\mathbb{F}_p^*:=\mathbb{F}_p\setminus\{0\}$ denote the multiplicative group of $\mathbb{F}_p$. Additionally, let $\sigma\in \mathbb{F}_p^*$ be of order $\mathrm{ord}(\sigma)=\ell$. Then, ${\mathcal{G}'}=\{\sigma^{0},\ldots,\sigma^{\ell-1}\} $ forms a subgroup of $\mathbb{F}_p^*$. We now form $T=|\mathbb{F}_p^*|/\ell$ cosets $[\tau_i]_{\mathcal{G}'}$ of size $\ell$ by choosing $\tau_i$ in the following way. First, let $\tau_0=1$. Then, for $i\in \{1,2,\ldots,T-1\}$, we consecutively choose
    $\tau_i \in \mathbb{F}_p^*\setminus\bigcup_{j=0}^{i-1} [\tau_j]_{\mathcal{G}'}.$
    This choice ensures that every $\tau_i$ belongs to a distinct coset. Now, form the matrix 
    \[
    \bm{M}=\begin{pmatrix}
        1&\sigma&\ldots &\sigma^{\ell-1}\\
        \sigma^{\ell-1}&1 &&\sigma^{\ell-2}\\
        \vdots& &\ddots &\vdots\\
        \sigma &\sigma^2&\ldots&1
    \end{pmatrix} \in \left(\mathbb{F}_p^*\right){}^{\ell \times \ell},
    \]
    as well as the matrix 
    \[\label{eq:tire_construction}
    \bm{\mathcal{H}}=\begin{pmatrix}
        \bm{1}_{\ell\times 1}&\tau_0\bm{M}&\ldots &\tau_{T-1}\bm{M}
    \end{pmatrix} \in \left(\mathbb{F}_p^*\right)^{\ell \times p}.
    \]
    Let $\bm{\mathcal{H}}_{j_1,:}$ and $\bm{\mathcal{H}}_{j_2,:}$ be two distinct rows of $\bm{\mathcal{H}}$, i.e., $j_1,j_2\in [\ell]$, $j_1\neq j_2$. Their difference vector $\bm{d} :=\bm{\mathcal{H}}_{j_1,:}-\bm{\mathcal{H}}_{j_2,:}$ can be written as
    $\bm{d} = \begin{pmatrix}
    0&\bm{d}^{(0)}&\bm{d}^{(1)}&\cdots&\bm{d}^{(T-1)}\\
    \end{pmatrix}\in \mathbb{F}_p^{p}$
    where for $i\in [T]$,
    \begin{align*}
    \bm{d}^{(i)} &= \tau_i\cdot(\bm{M}_{j_1}-\bm{M}_{j_2})
    :=\begin{pmatrix}
        d^{(i)}_0&d^{(i)}_1&\cdots&d^{(i)}_{\ell-1}\\
    \end{pmatrix}\in \mathbb{F}_p^{\ell}.
    \end{align*}
    For $x\in[\ell]$, we write
    \begin{align*}d^{(i)}_x &=\tau_i \cdot(\sigma^{\ell-j_1+x}-\sigma^{\ell-j_2+x})
    =\tau_i\cdot(\sigma^{-j_1}-\sigma^{-j_2})\cdot\sigma^{\ell+x}.
    \end{align*}
Thus, one can see that 
\[\left\{d^{(i)}_x:x\in [\ell]\right\}\equiv [\tau_i(\sigma^{-j_1}-\sigma^{-j_2})]_{\mathcal{G}'}.\]
 Therefore, the sets $\{d^{(i)}_x:x\in [\ell]\}$ for $i\in[T]$ yield the $T$ distinct cosets permuting all elements in $\mathbb{F}_p^*$. Together with the first $0$ element, $\bm{d}$ forms a permutation vector of $\mathbb{F}_p$. Performing an arbitrary partition of
$\bm{\mathcal{H}}$ into ${\bm{\mathcal{H}}=\begin{pmatrix}
        \bm{\mathcal{H}}_1^{\mathsf{T}}&
        \bm{\mathcal{H}}_2^{\mathsf{T}}
    \end{pmatrix}^{\mathsf{T}}}$ with $\bm{\mathcal{H}}_1 \in {(\mathbb{F}_p^{*})}^{\ell_1\times p}$, $\bm{\mathcal{H}}_2 \in {(\mathbb{F}_p^{*})}^{\ell_2\times p}$, and $\ell_1+\ell_2=\ell$ and using Lemma~\ref{lemma:condition_for_HHT1}, we know that $\mathrm{Cyc}( \bm{\mathcal{H}}_1)\mathrm{Cyc}( \bm{\mathcal{H}}_2)^{\mathsf{T}}=\bm{1}$. In this work, $\ell$ is always an even number and we choose $\ell_1=\ell_2$.

    It remains to show that the Tanner graph of $\mathrm{Cyc}( \bm{\mathcal{H}})$ has girth at least $6$ which is fulfilled if the difference vector of any two rows of $\mathrm{Cyc}( \bm{\mathcal{H}})$ is multiplicity free \cite{Hagiwara2007QQLDPC}. This condition is met because the difference vectors of the rows of $\bm{\mathcal{H}}$ are permutation vectors.
\end{proof}
We use the matrices $\bm{H}_1$ and $\bm{H}_2$ from the proof of Theorem~\ref{theorem:QCcosetconstruction} to construct QSCs as described in Sec.~\ref{sec:code_and_decoder}. The parameters of the constructed codes are listed in Tab.~\ref{tab:QC_codes}
\begin{example}
   We consider an example for $p=7$. Choosing $\sigma=3$ yields the subgroup $\mathcal{G}'=[7]$. Note that $\mathrm{ord}(\sigma)=6$ and $T=1$. Then, choose $\tau_0=1\in [1]_{\mathcal{G}'}$. Since $T=1$, no more $\tau_i$ are required. Hence, we can construct the base matrix 
\[
\bm{\mathcal{H}}=\begin{pmatrix}
        \bm{1}_{\ell\times 1}&\tau_0\bm{M}
    \end{pmatrix}=\footnotesize \begin{pmatrix}
    1&1&3&2&6&4&5\\
    1&5&1&3&2&6&4\\
    1&4&5&1&3&2&6\\
    1&6&4&5&1&3&2\\
    1&2&6&4&5&1&3\\
    1&3&2&6&4&5&1
\end{pmatrix}.
\]
	
One can verify that the difference of any two rows, computed in $\mathbb{F}_p$, results in a permutation vector of $\mathbb{F}_p$. Note that we choose $T=1$ as it yields low\--rate codes, which are suitable for the current quantum channels with high noise level. We take the first three rows of $\bm{\mathcal{H}}$ to be $\bm{\mathcal{H}}_1$ and the last three rows to be $\bm{\mathcal{H}}_2$. Then, we obtain two binary matrices Cyc($\bm{\mathcal{H}}_1$) and Cyc($\bm{\mathcal{H}}_2$) $\in \mathbb{F}_2^{21\times 49}$, both with rank $19$.\footnote{Note that $\bm{H}$ is almost never of full rank. Therefore, the constructed codes naturally have an overcomplete set of stabilizers that will all be used for decoding. The same holds for the codes constructed in Sec.~\ref{sec:coded_from_FG}. The benefits and complexity of this approach are discussed in~\cite{miao2023quaternary}.} By appending an all-one column to them, we obtain the Q1 code.
\end{example}

\begin{table}[t]
    \centering
    \caption{QC Code constructed with ord($\sigma$)=$p-1$.}
    \vspace{-0.5em}
    \label{tab:QC_codes}
    \begin{tabular}{ccccccc}
    \toprule
   Code&$n$&$k$&rate&$p$&$\sigma$&$d$\\
         \hline
         Q1&50&12&0.24&7&3&6\\
         Q2&122&20&0.16&11&2&12\\
         Q3&170&24&0.14&13&2&14\\
         Q4&290&32&0.11&17&3&18\\
         Q5&362&36&0.10&19&3&20\\
         \bottomrule
    \end{tabular}
    \vspace{-1em}
\end{table}

\section{QLDPC Codes from Finite Geometries}\label{sec:coded_from_FG}
In~\cite{Aly08QEG}, a QLDPC code construction using \acp{FG} fulfilling \eqref{eq:HHT1} was proposed. We briefly review the construction and focus on concrete examples of the constructed codes and their decoding performance with CAMEL.

Consider a set $\mathcal{N}$ of $N$ points and a set $\mathcal{M}$ of $M$ lines constructed from a certain finite field. Details on the construction method can be found in~\cite{lin1983error,Kou01EG}. The points and lines form an \ac{FG} if the following conditions are satisfied for some fixed integers $\gamma\geq 2$ and $\rho\geq 2$:
\begin{enumerate}
    \item Each line passes through $\rho$ points,
    \item any two points are on exactly one line,\label{property:2}
    \item each point lies on $\gamma$ lines,   
    \item two lines are either parallel or intersect at one and only one point.
\end{enumerate}

In this work, we use the 2D-\acp{FG} based on finite fields $\mathbb{F}_{q}$ of characteristic $2$ with $q=2^s$. This yields codes with the best minimum distances and decoding performance among the codes we constructed using \acp{FG}. Then, we focus on the most famous examples of \acp{FG} which are the \acp{EG} and the \acp{PG}. A 2D-\ac{EG} consists of $N=q^2$ points and $M=q^2+q$ lines. Moreover, we have $\rho=q$ and $\gamma=q+1$. A 2D-\ac{PG} contains $N=q^2+q+1$ points and $M=q^2+q+1$ lines where $\rho=q+1$ and $\gamma=q+1$. 

We index the points in an \ac{FG} from $0$ to $N-1$. For each line in an \ac{FG}, indexed from $0$ to $M-1$, define an \textit{incidence vector} $\vec{a}\in \mathbb{F}_2^{N}$ as follows: for $i\in [N]$, $a_i=1$ if the point $i$ is on the line and $a_i=0$ otherwise.

Now, form a binary \ac{PCM} $\bm{H}\in \mathbb{F}_2^{N\times M}$ whose columns consist of the incidence vectors of all lines in the \ac{FG}. It follows that $\bm{H}$ has a row weight of $\gamma$ and a column weight of $\rho$. Moreover, it was shown in~\cite{Kou01EG} that the minimum distance of the binary linear code defined by $\bm{H}$ is lower bounded by $\rho$. 

For this binary \ac{PCM} $\bm{H}\in \mathbb{F}_2^{N\times M}$, it is easy to show that 
$\vec{H} \vec{H}^{\mathsf{T}}=\bm{1}$ using condition (\ref{property:2}) and with our assumption that $q=2^s$.\footnote{ We have to point out a mistake in~\cite{Aly08QEG} where the columns of $\bm{H}$ are the incidence vectors of the lines which \textit{do not pass through origin} instead of \textit{all the lines} in the geometry as we do. We can verify that the former does not yield a matrix $\bm{H}$ that fulfills \eqref{eq:HHT1} by looking at any two non-origin points $i$ and $j$ located on a line $\lambda$ which passes through the origin. In our construction, it means that the $\lambda$-th column is the only column where the two rows $\bm{H}_{i,:}$ and $\bm{H}_{j,:}$ have an overlapping one as two points are on exactly one line. Therefore, $\bm{H}_{i,:}\bm{H}_{j,:}^{\mathsf{T}}=1$. Thus, removing column $\bm{H}_{:,\lambda}$ leads to $\bm{H}_{i,:}\bm{H}_{j,:}^{\mathsf{T}}=0$.} To construct QLDPC codes with relatively low rates, we choose ${\bm{H}_X=\bm{H}_Z = \begin{pmatrix}
    \vec{H} \mid \bm{1}
\end{pmatrix} \in \mathbb{F}_2^{N\times (M+1)}}$.\footnote{This unfortunately introduces some short cycles of length 4 between $\bm{H}_X$ and $\bm{H}_Z$. However, thanks to the other good properties of the FG codes, the performance degradation is acceptable.} One can verify that the minimum distance of the code defined by the PCM $\begin{pmatrix}
    \vec{H} \mid \bm{1}
\end{pmatrix}$ is the same as the minimum distance of the code defined by $\bm{H}$, which is lower bounded by $\rho$. Therefore, for codes constructed from the 2D-\ac{FG}, the minimum distance is approximately $\sqrt{n}$.

\begin{table}[t]
    \centering
    \caption{QLDPC codes from 2D \ac{EG} of finite field of characteristic 2.}
    \label{tab:FG_codes}
    \begin{tabular}{cccccc}
    \toprule
         Code&$n$&$k$&rate&$s$&$d$\\
         \hline
         E1&7&1&0.14&1&3\\
         E2&21&3&0.14&2&5\\
         E3&73&19&0.26&3&9\\ 
    E4&273&111&0.41&4&17\\ 
         E5&1057&571&0.54&5&$33$\\ 
         \bottomrule
    \end{tabular}
\end{table}
\begin{figure*}[t]
  \begin{minipage}[b]{6.2cm}
    \centering
    \begin{tikzpicture}
\pgfplotsset{grid style={gray!70}}
\pgfplotsset{every tick label/.append style={font=\footnotesize}}
\begin{axis}[%
name=ax1,
width=\textwidth+0.4cm,
height=6cm,
xmin=0.01,
xmax=0.12,
xmode=log,
ymode=log,
ymin=1e-5,
ymax=1,
xminorgrids,
yminorticks,
xmajorgrids,
ymajorgrids,
axis background/.style={fill=white, mark size=1.5pt},
ylabel={Frame Error Rate (FER)},
ylabel style={yshift=-0.18cm},
label style={font=\small},
xlabel={Depolarizing probability $\varepsilon$},
enlargelimits=false,
legend cell align={left},
legend style={at={(1,0)},anchor=south east,fill opacity=1, text opacity = 1,legend columns=1, row sep = 0pt,font=\footnotesize}
]

\addplot [color=EG4, line width=1pt,mark=diamond, mark size = 2pt]table[row sep=crcr]{
0.08 0.453725\\
0.07 0.29009\\
0.06 0.106422\\
0.05 0.0334999\\
0.04 0.00446716\\
0.03 0.000372533\\
0.02 6.22196e-06\\
0.018 2.05566e-06\\
0.016 7.66665e-07\\
0.014 2.16666e-07\\
0.012 4.99999e-08\\
};
\addlegendentry{CAMEL}

\addplot [color=myred, line width=1pt,mark=o, mark size=1pt]table[row sep=crcr]{
0.1 0.811966\\
0.09 0.682857\\
0.08 0.49325\\
0.07 0.291399\\
0.06 0.131266\\
0.05 0.0667531\\
0.04 0.0384419\\
0.03 0.0247482\\
0.02 0.0198195\\
0.01 0.00991643\\
}; 
\addlegendentry{BP}

\addplot [color=black, dashed, line width=1pt,mark=none]table[row sep=crcr]{
0.1 0.854809\\
0.09 0.70679\\
0.08 0.490308\\
0.07 0.262483\\
0.06 0.106832\\
0.05 0.0296812\\
0.04 0.0047212\\
0.03 0.000298626\\
0.02 5.1109e-06\\
}; 
\addlegendentry{BP (GA)}

\addplot [color=kit-green70, line width=1pt,mark=triangle, mark size=2pt]table[row sep=crcr]{
0.1 0.960422\\
0.09 0.909548\\
0.08 0.781182\\
0.07 0.629295\\
0.06 0.384442\\
0.05 0.198142\\
0.04 0.0834434\\
0.03 0.0348705\\
0.02 0.0207721\\
0.01 0.00906228\\
}; 
\addlegendentry{BP2}

\addplot [color=kit-green70, dashed, line width=1pt,mark=triangle, mark size=2pt, mark options = {solid}]table[row sep=crcr]{
0.05 0.163055\\
0.04 0.0475063\\
0.03 0.00641463\\
0.02 0.000208198\\
0.015 1.81667e-05\\
}; 
\addlegendentry{BP2 (GA)}

\addplot [color=KITyellow!70, line width=1pt,mark=square, mark size=2pt]table[row sep=crcr]{
0.08 0.7496\\
0.07 0.5700000000000001\\
0.06 0.35819999999999996\\
0.05 0.17420000000000002\\
0.04 0.07740000000000002\\
0.03 0.03300000000000003\\
0.02 0.017800000000000038\\
0.01 0.00539999999999996\\
}; 
\addlegendentry{BP2-OSD}

\end{axis}
\end{tikzpicture}  
    \caption{FER vs. depolarizing probability $\varepsilon$ curves when decoding an E4 $[[273,111,17]]$ QLDPC code from an \ac{EG} using different decoding algorithms.}
\label{fig:EG_code_compare_decoder}
  \end{minipage}
  \hspace{0.15cm}
  \begin{minipage}[b]{5.5cm}
    \hspace{-0.3cm}
    \begin{tikzpicture}
\pgfplotsset{grid style={gray!70}}
\pgfplotsset{every tick label/.append style={font=\footnotesize}}
\begin{axis}[%
name=ax1,
width=6.8cm,
height=6cm,
xmin=0.009,
xmax=0.12,
xmode=log,
ymode=log,
ymin=1e-5,
ymax=1,
xminorgrids,
yminorticks,
xmajorgrids,
ymajorgrids,
axis background/.style={fill=white, mark size=1.5pt},
yticklabels={},
xlabel={Depolarizing probability $\varepsilon$},
label style={font=\small},
enlargelimits=false,
legend cell align={left},
legend style={at={(1,0)},anchor=south east,fill opacity=1, text opacity = 1,legend columns=1, row sep = 0pt,font=\footnotesize}
]

\addlegendentry{Q1}
\addlegendentry{Q2}
\addlegendentry{Q3}
\addlegendentry{Q4}
\addlegendentry{Q5}
\addlegendentry{R1}

\addplot [color=QC1, line width=1pt,mark=o]table[row sep=crcr]{
0.1 0.43843\\
0.09 0.354422\\
0.08 0.268494\\
0.07 0.196937\\
0.06 0.131459\\
0.05 0.0736777\\
0.04 0.0376855\\
0.03 0.0157286\\
0.02 0.00427124\\
0.01 0.000494445\\
};
\addplot [color=QC2, line width=1pt,mark=square]table[row sep=crcr]{
0.1 0.712025\\
0.09 0.581267\\
0.08 0.440953\\
0.07 0.286381\\
0.06 0.168857\\
0.05 0.0826277\\
0.04 0.0275881\\
0.03 0.0066705\\
0.02 0.000794946\\
0.01 2.48878e-05\\
};

          \addplot [color=QC3, line width=1pt,mark=triangle]table[row sep=crcr]{
0.1 0.831382\\
0.09 0.709611\\
0.08 0.5392\\
0.07 0.39951\\
0.06 0.234154\\
0.05 0.0976452\\
0.04 0.0271457\\
0.03 0.00470537\\
0.02 0.000359564\\
0.01 5.7777e-06\\
};

\addplot [color=QC4, line width=1pt,mark=diamond]table[row sep=crcr]{
0.1 0.981308\\
0.09 0.951662\\
0.08 0.850136\\
0.07 0.659664\\
0.06 0.423077\\
0.05 0.169758\\
0.04 0.0343885\\
0.03 0.00405704\\
0.02 0.000113335\\
0.01 4.44438e-07\\
};

\addplot [color=QC5, line width=1pt,mark=star]table[row sep=crcr]{
0.1 1\\
0.09 0.975806\\
0.08 0.930233\\
0.07 0.829099\\
0.06 0.524409\\
0.05 0.249027\\
0.04 0.0571059\\
0.03 0.004313\\
0.02 0.0000613\\
0.016 6.33327e-06\\
};

\addplot [color=QC5, dashed, line width=1pt,mark=star]table[row sep=crcr]{
    0.08995451958872763 0.8074878458635921\\
0.07910314675722364 0.6107653946066945\\
0.06956079422694146 0.4151195836467042\\
0.06140126800498785 0.26462050618904925\\
0.049872983923963235 0.1197102765411727\\
0.03990113244336744 0.05531797752399404\\
0.029935772947204904 0.02295631924236914\\
0.01989988536895992 0.006615389893015433\\
0.009924666630948488 0.0008435452308265797\\
};

\end{axis}
\end{tikzpicture}  
    \caption{FER vs. depolarizing probability $\varepsilon$ curves for the quasi-cyclic QLDPC codes shown in Tab.~\ref{tab:QC_codes} and the reference code R1.}
\label{fig:FER1}
  \end{minipage}
  \hspace{0.15cm}
  \begin{minipage}[b]{5.5cm}
    \hspace{-0.4cm}
    \begin{tikzpicture}
\pgfplotsset{grid style={gray!70}}
\pgfplotsset{every tick label/.append style={font=\footnotesize}}
\begin{axis}[%
name=ax1,
width=6.8cm,
height=6cm,
xmin=0.009,
xmax=0.12,
xmode=log,
ymode=log,
ymin=1e-5,
ymax=1,
xminorgrids,
yminorticks,
xmajorgrids,
ymajorgrids,
axis background/.style={fill=white, mark size=1.5pt},
yticklabels={},
xlabel={Depolarizing probability $\varepsilon$},
label style={font=\small},
legend cell align={left},
legend style={at={(1,0)},anchor=south east,fill opacity=1, text opacity = 1,legend columns=1, row sep = 0pt,font=\footnotesize}
]
\addlegendentry{E1}
\addlegendentry{E2}
\addlegendentry{E3}
\addlegendentry{E4}
\addlegendentry{E5}

\addplot [color=EG1, line width=1pt,mark=o]table[row sep=crcr]{
0.1 0.119632\\
0.09 0.0950141\\
0.08 0.0782053\\
0.07 0.062368\\
0.06 0.0466216\\
0.05 0.0348514\\
0.04 0.0242231\\
0.03 0.0126104\\
0.02 0.00631675\\
0.01 0.00158278\\};

\addplot [color=EG2, line width=1pt,mark=square]table[row sep=crcr]{
0.1 0.182203\\
0.09 0.140457\\
0.08 0.09632\\
0.07 0.0732004\\
0.06 0.0494415\\
0.05 0.0333149\\
0.04 0.0195696\\
0.03 0.00839071\\
0.02 0.00286992\\
0.01 0.00036938\\
};

\addplot [color=EG3, line width=1pt,mark=triangle]table[row sep=crcr]{
0.1 0.482372\\
0.09 0.347575\\
0.08 0.252094\\
0.07 0.165932\\
0.06 0.108468\\
0.05 0.0459823\\
0.04 0.0196015\\
0.03 0.00466688\\
0.02 0.000632926\\
0.01 0.000019\\
};
\addplot [color=EG4, line width=1pt,mark=diamond]table[row sep=crcr]{
0.08 0.453725\\
0.07 0.29009\\
0.06 0.106422\\
0.05 0.0334999\\
0.04 0.00446716\\
0.03 0.000372533\\
0.02 6.22196e-06\\
};
\addplot [color=EG5, line width=1pt,mark=star]table[row sep=crcr]{
0.08 1\\
0.07 1\\
0.06 0.979592\\
0.05 0.798107\\
0.04 0.275081\\
0.03 0.0132561\\
0.03 0.0150717\\
0.028 0.00429531\\
0.026 0.00159324\\
0.024 0.000300749\\
0.02 0.000007\\
};

\addplot [color=EG5, dashed, line width=1pt,mark=star]table[row sep=crcr]{
0.030003338898163603  0.061017519056170144\\
0.024994991652754595  0.011433569030716717\\
0.01998664440734558  0.0007844142136982535\\
0.018003338898163607  0.0002196939927297685\\
0.016  0.00005204315734195076\\
0.013956594323873122  0.000011338235012178496\\
}; 
\addlegendentry{R2}

\end{axis}
\end{tikzpicture}  
    \caption{FER vs. depolarizing probability $\varepsilon$ curves for QLDPC codes constructed using \acp{EG} shown in Tab.~\ref{tab:FG_codes} and the reference code R2.}
\label{fig:FER2}
  \end{minipage}
\end{figure*}

The parameters of the exemplary codes constructed from \acp{EG} are listed in Tab.~\ref{tab:FG_codes}. As for the \ac{PG} codes, assume 
that a QLDPC code constructed from an \ac{EG} using a finite field has parameters $[[n,k,d]]$. Then, the QLDPC code constructed from a \ac{PG} using the same field has parameters $[[n+1,k+1,d+1]]$. Hence, we do not list them explicitly.

\section{Numerical Results}\label{sec:simulation}

We assess the proposed scheme using Monte Carlo simulations over the quantum depolarizing channel where the three types of Pauli errors occur with equal probability $\varepsilon/3$. 
At least 300 frame errors are collected to obtain the \ac{FER} for each data point. 
All BP decoding paths use sum-product algorithm with $15$ iterations and a flooding schedule.

First, we highlight the importance of decoding using our proposed ensemble decoder. Figure~\ref{fig:EG_code_compare_decoder} depicts the decoding results using different decoding algorithms for the E4 code as an example. When decoded with a single BP decoder, the decoding performance is poor due to numerous short cycles of length 4. However, when using ensemble decoding in CAMEL as described in Sec.~\ref{sec:code_and_decoder}, the performance improves by orders of magnitude. Additionally, no error floor is observed when simulating at a low FER of $10^{-7}$. Moreover, as a benchmark, we evaluate a genie-aided (GA) BP decoder which is fed the correct value of $\mathsf{v}_n$. The ensemble BP decoding has almost the same performance as the genie-aided version. As the reliability of the guessed value is set to infinity, corresponding to a probability of $1$, the paths with the wrong guess usually either fail to converge or will find a high-weight error estimate. For completeness, we also plot the decoding results using a pair of binary BP (BP2) decoders and its genie-aided version where the correct value of $\mathsf{v}_n$ is fed to both decoders. They are outperformed by their respective quaternary counterpart. Furthermore, we decode the E4 code using the BP2-OSD decoder with combination sweep (CS) strategy and an order of $42$, which is implemented in~\cite{roffe2020decoding}. Further increasing the order of OSD does not improve the performance noticeably. The decoding gain achieved by the BP-OSD decoder is limited compared to the genie-aided BP2 decoder. For other constructed codes, we observe similar results.

We also compare the performance of the QLDPC codes constructed in Sec.~\ref{sec:coded_from_QC} and Sec.~\ref{sec:coded_from_FG} using the proposed ensemble BP decoder with existing codes as depicted in Fig.~\ref{fig:FER1} and Fig.~\ref{fig:FER2}, respectively. Our Q5 code outperforms the $[[400,16]]$ hypergraph product code (R1) with BP-OSD decoder~\cite{roffe2020decoding} and our E5 code outperforms the $[[800,400]]$ bicycle code (R2) using the modified non-binary decoder with enhanced feedback~\cite{Babar15,Wang12enhanced}. Note that both reference codes use a significantly more complex decoding algorithm and have lower code rates than our constructed codes. We conclude that CAMEL achieves great decoding performance with low decoding latency.

\section{Conclusion}\label{sec:conclusion}
In this paper, we proposed a novel joint code and decoder design for QLDPC codes with good quarternary BP decoding performance without the need for any major modification of the BP decoder or any post-processing steps. Simulation results show a significant improvement compared to conventional BP decoding.

Future work includes investigating the generalization of the proposed CAMEL scheme to provide codes with improved properties. For example, constructing codes that are more degenerate, enabling joint decoding of the circuit-level noise, and constructing codes with local qubit connectivity.

\IEEEtriggeratref{2}

\IEEEtriggeratref{12}

\end{document}